\documentclass[12pt,a4paper,reqno]{article}
\usepackage{graphics}
\usepackage{epsfig}

\usepackage{amsmath,amsthm,amssymb}
\newtheorem{theorem}{Theorem}[section]
\newtheorem{corollary}[theorem]{Corollary}

\newtheorem{proposition}[theorem]{Proposition}
\newtheorem{example}[theorem]{Example}
\theoremstyle{definition}
\newtheorem{definition}[theorem]{Definition}
\theoremstyle{remark}

\numberwithin{equation}{section}

\begin{document}
\title{Relative two-weight $\mathbb{Z}_{2}\mathbb{Z}_{4}$-additive Codes
\thanks{The first author would like to thank the Department of Science and Technology (DST), New Delhi, India for their financial support
in the form of INSPIRE Fellowship
(DST Award Letter No.IF130493/DST/INSPIRE Fellowship / 2013/(362) Dated: 26.07.2013) 
to carry out this work.}
}
\author{N. Annamalai\\
Research Scholar\\
Department of Mathematics\\
Bharathidasan University\\
Tiruchirappalli-620 024, Tamil Nadu, India\\
{Email: algebra.annamalai@gmail.com}
\bigskip\\
C.~Durairajan\\
Assistant Professor\\
Department of Mathematics\\ 
School of Mathematical Sciences\\
Bharathidasan University\\
Tiruchirappalli-620024, Tamil Nadu, India\\
{Email: cdurai66@rediffmail.com}
\hfill \\
\hfill \\
\hfill \\
\hfill \\
{\bf Proposed running head:} Relative two-weight $\mathbb{Z}_{2}\mathbb{Z}_{4}$-additive Codes}
\date{}
\maketitle

\newpage

\vspace*{0.5cm}
\begin{abstract}In this paper, we study a relative two-weight $\mathbb{Z}_{2}\mathbb{Z}_{4}$-additive codes. It is shown that the Gray image
of a two-distance $\mathbb{Z}_{2}\mathbb{Z}_{4}$-additive code is a binary two-distance code and that the Gray image of a relative two-weight
$\mathbb{Z}_{2}\mathbb{Z}_{4}$-additive code, with nontrivial binary part, is a linear binary relative two-weight code. 
The structure of relative two-weight $\mathbb{Z}_{2}\mathbb{Z}_{4}$-additive codes are described. 
Finally, we discussed permutation automorphism group of a  $\mathbb{Z}_{2}\mathbb{Z}_{4}$-additive codes.
\end{abstract}
\vspace*{0.5cm}
%
\vspace{0.5cm}
\vspace{1.5cm}

\noindent
Corresponding author:\\ 
\\
 \hspace*{1cm} 
Dr. C. Durairajan\\
\hspace*{1cm}
Assistant Professor\\
\hspace*{1cm}
Department of Mathematics \\
\hspace*{1cm}
Bharathidasan University\\
\hspace*{1cm}
Tiruchirappalli-620024, Tamil Nadu, India\\
\hspace*{1cm}
E-mail: cdurai66@rediffmail.com
\newpage

\newpage




\section{Introduction}

\quad A q-ary code of length $n$ over a finite field $\mathbb{F}_{q}$ of size $q$ is a subset 
of $\mathbb{F}_{q}^n.$ If $q = 2,$ then the code is called a binary code. If it is a subspace 
of $\mathbb{F}_{q}^{n},$ then it is called a $q$-ary linear code. 

In \cite{bon}, Bonisoli has determined the structure of linear one-weight codes over 
finite fields and proved that every equidistant linear code is a sequence of Simplex codes. In \cite{car}, Carlet studied
linear one-weight codes over $\mathbb{Z}_{4}$ where the weight used  is the Lee weight. In 2000,
Wood \cite{wood} determined the structure of linear one-weight codes over $\mathbb{Z}_{m}$ for various weights.

In \cite{del}, Delsarte gave Fundamental results on additive codes. He showed that any abelian binary propelinear code has the form
$\mathbb{Z}_{2}^{\gamma}\times \mathbb{Z}_{4}^{\delta}$ for some nonnegative integers
$\gamma$ and $\delta$ with $\gamma + \delta >0.$ Hence, it become important to study
additive codes in $\mathbb{Z}_{2}^{\alpha}\times \mathbb{Z}_{4}^{\beta}$ for some nonnegative 
integers $\alpha$ and $\beta.$ A $\mathbb{Z}_{2}\mathbb{Z}_{4}$-code is a subset of 
$\mathbb{Z}_{2}^{\alpha}\times\mathbb{Z}_{4}^{\beta}$ where $\alpha, \beta$ are nonnegative integers such that $\alpha+\beta>0.$ 
If it is a subgroup of $\mathbb{Z}_{2}^{\alpha}\times\mathbb{Z}_{4}^{\beta},$ then it is 
called an {\it additive code.}
Fundamental results on $\mathbb{Z}_{2}\mathbb{Z}_{4}$-additive codes, including the generator matrix, the existence and constructions
of self-dual codes and several different bounds can be found in \cite{sole} - \cite{df}.
The structure of one-weight $\mathbb{Z}_{2}\mathbb{Z}_{4}$-additive codes have been determined by Dougherty et al.\cite{dhl}. Relative one-weight
linear codes were introduced by Liu and Chen over finite fields \cite{liu} and \cite{lich}. Automorphism groups of Grassmann codes were discussed 
by Sudhir R. Ghorpade and Krishna V. Kaipa \cite{ghor}.

In this paper, we introduced a relative two-weight additive codes  over $\mathbb{Z}_{2}\mathbb{Z}_{4}$. We showed that the Gray image of a two-distance
 $\mathbb{Z}_{2}\mathbb{Z}_{4}$-code is a binary two-distance code. We proved that the Gray image of 
 a relative two-weight additive $\mathbb{Z}_{2}\mathbb{Z}_{4}$-code is a binary relative 
 two-weight  linear code. We described the structure of relative 
two-weight $\mathbb{Z}_{2}\mathbb{Z}_{4}$-additive codes. Finally, we discussed permutation 
automorphism group of a $\mathbb{Z}_{2}\mathbb{Z}_{4}$-additive codes.

\section{Preliminaries}

Throughtout in this paper we denote the calligraphic $\mathcal{C}$ as a code in 
$\mathbb{Z}_{2}^{\alpha}\times\mathbb{Z}_{4}^{\beta}$ and the 
standard $C$ as a code over the binary field $\mathbb{Z}_{2}$.

The Gray map from $\mathbb{Z}_{4}$ to $\mathbb{Z}_{2}^{2}$ is given by
\begin{equation*}
 \phi(0)=(0,0), \phi(1)=(0,1),\phi(2)=(1,1) \text{ and }\phi(3)=(1,0).
\end{equation*}
This map was defind by F. Gray in 1940 to avoid  large errors in transmiting signal by pulse code modulation.
Since the Lee weight of 0 is 0, 1 is 1, 2 is 2 and 3 is 1, implies $wt_L(i) = wt_H(\phi(i))$ for all 
$i \in \mathbb{Z}_{4}$ where $wt_L(x)$ is the Lee weight of $x$ and  $wt_H(x)$ is the Hamming 
weight of $x.$ In \cite{vera}, they gave the following theorem
\quad\begin{theorem} The Gray map $\phi$ is a distance preserving map from $\mathbb{Z}_{4}^{n}$ with Lee distance to $\mathbb{Z}_{2}^{2n}$ with Hamming distance.
\end{theorem}

Let $(u|v)\in \mathbb{Z}_{2}^{\alpha}\times\mathbb{Z}_{4}^{\beta}$ where 
$u \in \mathbb{Z}_{2}^{\alpha}$ and 
$ v \in \mathbb{Z}_{4}^{\beta},$ then in \cite{df1}, the author defined the Gray map 
$\Phi : \mathbb{Z}_{2}^{\alpha}\times\mathbb{Z}_{4}^{\beta}\rightarrow \mathbb{Z}_{2}^{\alpha+2\beta}$ as
$$\Phi((u|v))= (u|\phi(v))\text{ for all } (u|v)\in \mathbb{Z}_{2}^{\alpha}\times\mathbb{Z}_{4}^{\beta}$$
where $\phi(v) = (\phi(v_1),\phi(v_2), \cdots, \phi(v_\beta)) \in  \mathbb{Z}_{4}^{\beta}.$

Let $(u|v),(u'|v')\in \mathbb{Z}_{2}^{\alpha}\times\mathbb{Z}_{4}^{\beta},$ then the Hamming distance 
between $(u|v)$ and $(u'|v')$ is defined by 
$$d_{H}((u|v),(u'|v'))=wt_{H}(u-u'|v-v').$$

If $\mathcal{C}$ is a $\mathbb{Z}_{2}\mathbb{Z}_{4}$-additive code, then 
$(u|v)-(u'|v') \in \mathcal{C}$ for all $(u|v), (u'|v') \in \mathcal{C}$ and hence 
$wt_{H}((u|v)-(u'|v'))= wt_{H}(u-u'|v-v') = d((u|v), (u'|v')).$ Therefore, the minimum Hamming distance and the minimum Hamming weight of
$\mathbb{Z}_{2}\mathbb{Z}_{4}$-additive code $\mathcal{C}$ are the same. 
The Lee weight of $(u|v)$ is defined by
$$wt_L((u|v))=wt_{H}(u)+wt_{H}(\phi(v))$$ 
and the Lee distance between $(u|v)$ and $(u'|v')$ as
$$d_L((u|v),(u'|v'))=wt_{H}(u-u')+wt_{H}(\phi(v-v')).$$

From the Theorem 2.1, it is easy to see that the Gray map 
$\Phi: \mathbb{Z}_{2}^{\alpha}\times\mathbb{Z}_{4}^{\beta}\rightarrow \mathbb{Z}_{2}^{n}$ is an isometry from 
$\mathbb{Z}_{2}^{\alpha}\times\mathbb{Z}_{4}^{\beta} $ with Lee distance to $\mathbb{Z}_{2}^{n}$ 
with Hamming distance where $n=\alpha+2\beta.$

Throughtout in this correspondance, we denote the minimum Hamming distance of the code 
$\mathcal{C}$ by $d_{H}(\mathcal{C}),$ the minimum Lee distance of the code $\mathcal{C}$ by $d(\mathcal{C})$ and  the minimum 
Lee weight of code $\mathcal{C}$ by $wt(\mathcal{C}).$  We write Lee weight as weight.
\section{Relative two-weight Codes in  $\mathbb{Z}_{2}^{\alpha}\times\mathbb{Z}_{4}^{\beta}$}

 A nonzero code $\mathcal{C}$ in $ \mathbb{Z}_{2}^{\alpha}\times\mathbb{Z}_{4}^{\beta}$ 
 is called a $\it{one\text{-}weight\, code}$ if  all its nonzero
 codewords  have the same Lee weight.

 A nonzero  code $\mathcal{C}$ in $ \mathbb{Z}_{2}^{\alpha}\times\mathbb{Z}_{4}^{\beta}$ is called a $\it{two\text{-}weight \,code}$ if all its nonzero
 codewords have two different Lee weights.

\begin{definition}
 A nonzero additive code $\mathcal{C}$ in $ \mathbb{Z}_{2}^{\alpha}\times\mathbb{Z}_{4}^{\beta}$ is called a 
 {\it relative two-weight additive code } to the subcode $\mathcal{C}_{1}$ of $\mathcal{C}$
 if $\mathcal{C}_{1}$ and $\mathcal{C}\setminus \mathcal{C}_{1}$ are one-weight codes.
 A relative two-weight additive code $\mathcal{C}$ to the subcode $\mathcal{C}_{1}$
 with $wt(c_1)=m_{1}$ for some $m_{1}>0$ for all $ c_1 \neq 0 \text{ in } \mathcal{C}_{1}$ and $wt(c)=m$ for all 
 $c\in \mathcal{C}\setminus \mathcal{C}_{1}$ for some $m>0,$ is denoted as $\mathcal{C}(m_{1},m).$
\end{definition}
\begin{definition}
 A nonzero additive code $\mathcal{C}$ in $ \mathbb{Z}_{2}^{\alpha}\times\mathbb{Z}_{4}^{\beta}$ is called a 
 {\it two-distance code } if there exists a pair of
 distinct positive integers $m,m_{1}$ such that for any two distinct codewords $c, d\in\mathcal{C},$ $d(c,d)\in \{m,m_{1}\}.$
\end{definition}

It is easy to see that if $\mathcal{C}$ is an additive code, then $\mathcal{C}$ is a  relative two-weight code if and only if $\mathcal{C}$ is a 
two-distance code.
\begin{theorem}
 If $\mathcal{C}$ is a two-distance additive code in $ \mathbb{Z}_{2}^{\alpha}\times\mathbb{Z}_{4}^{\beta}$ with distance $m$  and $m_{1}$, then
 $\Phi(\mathcal{C})$ is a binary two-distance code with the same distance $m$ and $m_{1}.$
\end{theorem}
\begin{proof}
 Let $\Phi(c),\Phi(d)\in \Phi(\mathcal{C})$ with $\Phi(c)\neq \Phi(d)$ where $c,d\in \mathcal{C}.$
 Clearly $c\neq d.$ Since $\mathcal{C}$ is a  $\mathcal{C}(m_{1},m)$ code, $d(c,d)\in\{m,m_{1}\}.$ 
 Since the Gray map $\Phi$ is an isometry, we have  $$d_{H}(\Phi(c),\Phi(d))=d(c,d)\in\{m,m_{1}\}.$$
 Hence, $\Phi(\mathcal{C})$ is a binary two-distance code with the same distance $m$ and $m_{1}.$
\end{proof}

Let $(u_1|v_1),(u_2|v_2)\in \mathbb{Z}_{2}^{\alpha}\times\mathbb{Z}_{4}^{\beta}$ 
where $u_1,u_2\in \mathbb{Z}_{2}^{\alpha}$ and
$v_1,v_2\in \mathbb{Z}_{4}^{\beta}.$ Then the inner product between $(u_1|v_1)$ and 
$(u_2|v_2)$ is defined by
$$\langle (u_1|v_1), (u_2|v_2)\rangle=2\langle u_1,u_2\rangle+\langle v_1,v_2\rangle\in \mathbb{Z}_{4},$$
where $\langle x,y\rangle=\sum\limits_{i=1}^{n}x_{i}y_{i}$ and the computation of
$2\langle u_1,u_2\rangle$ are done in $\mathbb{Z}_{4}.$

Let $\mathcal{C}\subseteq\mathbb{Z}_{2}^{\alpha}\times\mathbb{Z}_{4}^{\beta}$ be an
additive code. We define the dual 
code $\mathcal{C}^{\perp}$ of $\mathcal{C}$ as
$$\mathcal{C}^{\perp}=\{(x|y)\in \mathbb{Z}_{2}^{\alpha}\times\mathbb{Z}_{4}^{\beta} \mid \langle (x|y), (u|v)\rangle=0 \, \text{for all}\, (u|v) \in \mathcal{C}\}.$$


 The following examples show that if $\mathcal{C}$ is a relative two-weight additive code, 
 then ${\mathcal{C}^\perp}$ need not be so.

\begin{example}
 Let $\mathcal{C}=\langle (0|1)\rangle$ be an additive code in $\mathbb{Z}_{2}^{1}\times\mathbb{Z}_{4}^{1}.$
 Then $\mathcal{C}$ is a $\mathcal{C}(2,1)$  relative two-weight additive code to the subcode  $\mathcal{C}_{1}$ where 
 $\mathcal{C}_{1}= \langle(0|2)\rangle$ and its dual code $\mathcal{C}^{\perp}=\langle(1|0)\rangle$ is not a relative two-weight additive code.
\end{example}

 \begin{example}
 Let $\mathcal{C}$ be an additive code in $\mathbb{Z}_{2}^{2}\times\mathbb{Z}_{4}^{2}$ with generator matrix
 \begin{equation*}
G=
 \begin{bmatrix}
  1&0&\vline&1&1\\
  1&1&\vline&3&1
 \end{bmatrix}.
 \end{equation*}
   Then $\mathcal{C}$ is a relative two-weight additive code $\mathcal{C}(4,3)$ to the subcode $\mathcal{C}_{1}=\langle (11|31)\rangle$ and its
  dual code $\mathcal{C}^{\perp}=\langle (10|02),(00|22)\rangle$ which is a relative two-weight additive
  code with same weights $4$ and $3.$ Therefore,  both $\mathcal{C}$ and $\mathcal{C}^{\perp}$ 
  are  relative two-weight additive codes.
 \end{example}

 Let $\mathcal{C}$ be a relative two-weight additive  code $\mathcal{C}(m_{1},m)$ with respect to a 
 subcode $\mathcal{C}_{1}.$ Let $c=(u|v)\in \mathcal{C}_{1}.$ Since $\mathcal{C}_{1}$ is an additive one-weight subcode, 
 $c+c=(0|2v)\in \mathcal{C}_{1}.$ Since each coordinate of $2v$ in $\mathbb{Z}_4^\beta$ is either 0 or 2,
 implies $wt(2v)$ is even and hence $wt(0|2v)$ is even. Since $\mathcal{C}_{1}$ is an additive one-weight $m_1$
 code, $wt(c) = wt(2c)$ and $wt(c)$ is even for all nonzero codewords in $\mathcal{C}_{1}.$
Therefore, $m_{1}$ is even. Thus, we have

\begin{theorem}\label{a}
 Let $\mathcal{C}$ be an additive code in $\mathbb{Z}_{2}^{\alpha}\times\mathbb{Z}_{4}^{\beta}.$ If $\mathcal{C}(m_{1},m)$ is a relative two-weight
 additive code, then $m_{1}$ must be even.
\end{theorem}

\begin{theorem}
 Let $\mathcal{C}$ be a $\mathcal{C}(m_{1},m)$ code  in 
 $\mathbb{Z}_{2}^{\alpha}\times\mathbb{Z}_{4}^{\beta},$
 then the  Gray image $\Phi(\mathcal{C})$ of $\mathcal{C}$  is a binary relative two-weight code 
 $\Phi(\mathcal{C})(m_1,m)$in 
 $\mathbb{Z}_{2}^{\alpha+2\beta}.$
\end{theorem}
\begin{proof}
 Let $ x \in \Phi(\mathcal{C})\setminus \Phi(\mathcal{C}_{1}),$ then there exists $c\in \mathcal{C}\setminus\mathcal{C}_{1}$ such that $x=\Phi(c).$
 Since $wt_{H}(\Phi(c))=wt(c)$ for all $c\in \mathcal{C}.$ Therefore, $wt(x)=wt(\Phi(c))=wt(c)=m.$ 
 Let $x\in \Phi(\mathcal{C}_{1})$ with $x\neq 0,$ then there exists  $0\neq c_{1}\in \mathcal{C}_{1}$ such that $x=\Phi(c_{1}).$ Therefore, $wt(x)=wt(\Phi(c_{1}))=wt(c_{1})=m_{1}.$
 Hence $\Phi(\mathcal{C})(m_{1},m)$ is a binary relative two-weight code in  $\mathbb{Z}_{2}^{\alpha+2\beta}.$
\end{proof}

\begin{theorem}
 Let $\mathcal{C}$ be a relative two-weight additive code $\mathcal{C}(m_{1},m)$ in
 $\mathbb{Z}_{2}^{\alpha}\times\mathbb{Z}_{4}^{\beta}.$ Then for any positive integer $t,$ there exists a relative two-weight additive code
 $\mathcal{D}(tm_{1},tm)$ in $\mathbb{Z}_{2}^{t\alpha}\times\mathbb{Z}_{4}^{t\beta}.$
\end{theorem}
\begin{proof}
Let $\mathcal{C}$ be a relative two-weight additive code $\mathcal{C}(m_{1},m)$ to the subcode $\mathcal{C}_{1}.$
 Define $\mathcal{D}=\left\{(\underbrace{x\cdots x}_{\text{t times}}|\underbrace{y\cdots y}_\text{t times}) \mid (x|y)\in \mathcal{C}\right\}
 \subseteq\mathbb{Z}_{2}^{t\alpha}\times\mathbb{Z}_{4}^{t\beta}$
and 
$\mathcal{D}_{1}=\left\{(\underbrace{x\cdots x}_\text{t times}|\underbrace{y\cdots y}_\text{t times})\in\mathcal{D} \mid (x|y)\in \mathcal{C}_{1}\right\}.$

 Clearly, $\mathcal{D}_{1}\subseteq \mathcal{D}$ is an additive code in $\mathbb{Z}_{2}^{t\alpha}\times\mathbb{Z}_{4}^{t\beta}.$ 
Let $(\underbrace{x\cdots x}_{t \,\,\text{times}}|\underbrace{y\cdots y}_{t\,\,\text{times}})\in \mathcal{D}\setminus \mathcal{D}_{1}.$
 Then $(x|y)\in \mathcal{C}\setminus\mathcal{C}_{1},$  $wt(x|y) = m$ and hence $wt(u|v) = tm.$ 
  
  Let 
  $(\underbrace{x\cdots x}_\text{t times}|\underbrace{y\cdots y}_\text{t times})\in \mathcal{D}_{1}.$  
 Then $(x|y)\in \mathcal{C}_{1}.$ Since $\mathcal{C}_{1}$ is a one-weight code,  $wt(x|y) = m_1$ and hence $wt(u|v) = tm_1.$ 
  Therefore, $\mathcal{D}(tm_{1},tm)$ is a relative
 two-weight additive code to the subcode $\mathcal{D}_{1}.$
 \end{proof}
 
 \begin{theorem}
 Let $\mathcal{C}$ be an additive code in $\mathbb{Z}_{2}^{\alpha}\times\mathbb{Z}_{4}^{\beta}.$
 Then the weights of all codewords of $\mathcal{C}$
 are even if and only if $(\textbf{1}_{\alpha}|\textbf{2}_{\beta})\in \mathcal{C}^{\perp}$ where 
 $\textbf{1}_{\alpha}= (1,1, \cdots, 1) \in \mathbb{Z}_{2}^{\alpha} \text{ and } \textbf{2}_{\beta} = (2,2, \cdots, 2)\in 
 \mathbb{Z}_{4}^{\beta}.$
\end{theorem}
\begin{proof}
 Let $(u|v)\in \mathcal{C}$ where 
 $u\in \mathbb{Z}_{2}^{\alpha}$ and $v\in \mathbb{Z}_{4}^{\beta}.$ 
 Let us take $u=(u_1,\cdots,u_{\alpha})$ and $v=(v_1,\cdots,v_{\beta}),$ then by the definition of 
 inner product,
  $$\langle (\textbf{1}_{\alpha}|\textbf{2}_{\beta}), (u|v)\rangle=\sum_{i=1}^{\alpha}2u_{i}+\sum_{j=1}^{\beta}2v_{j}.$$
  It is easy to see that $\langle (\textbf{1}_{\alpha}|\textbf{2}_{\beta}), (u|v)\rangle=0$ if and  only if 
  $wt((u|v))$ is even. Therefore, the weight of all codewords in $\mathcal{C}$ are even iff 
  $(\textbf{1}_{\alpha}|\textbf{2}_{\beta})\in \mathcal{C}^{\perp}.$
\end{proof}

Combining this Theorem and Theorem \ref{a}, we have

\begin{corollary}
 Let $\mathcal{C}(m_{1},m)$ be a relative two-weight additive code in $\mathbb{Z}_{2}^{\alpha}\times\mathbb{Z}_{4}^{\beta}.$
 Then the weights of all codewords of $\mathcal{C}$
 are even if and only if $(\textbf{1}_{\alpha}|\textbf{2}_{\beta})\in \mathcal{C}^{\perp}.$
\end{corollary}

\section{The Structure of Relative two-weight Additive Codes}

 In a ring, the element $x$ is called {\it unit } if there exists $y$ such that $xy=1.$ The nonzero element $a$ 
 is said to be a {\it zero divisor } if there exists a nonzero element $b$ such that $ab=0.$

\begin{theorem}
 Let $G=
 \begin{pmatrix}
  A\\
  \hline
  G_{1}
 \end{pmatrix}$ be the generator matrix of a relative two-weight additive code $\mathcal{C}(m_{1},m)$ to the subcode $\mathcal{C}_{1}$
 in $\mathbb{Z}_{2}^{\alpha}\times\mathbb{Z}_{4}^{\beta}$ where $G_{1}$ is a generator matrix of the subcode $\mathcal{C}_{1}.$
 If $c=(u|v)$ is a row of $G_{1},$ then the number of units in $v$ is either $0$ or
 $\frac{m_{1}}{2}.$ 
\end{theorem}
\begin{proof}
 Let $c=(u|v)$ be a row of $G_{1}.$ Then $c+c=(0|2v).$ If $c+c=0,$ then $v$ contains no units. 
 If $c+c\neq 0,$ then $wt(2c)=wt(0|2v)=wt_{H}(0)+wt_{L}(2v).$ Since $wt(c)=m_{1}$ 
 for all $c\in \mathcal{C}_{1},  wt_{L}(2v)=m_{1}.$
 Since the number of unit places in $v$ is the same as number of nonzero coordinates in $2c$
  and nonzero coordinates of $2c$ are 2, the 
  $wt(2c) = 2 \times \text{( number of unit places of $c$)}.$
  Therefore, $2 \times \text{number of units of } v = m_{1}$ and hence 
  the number of unit coordinates in $v$ is $\frac{m_{1}}{2}.$
\end{proof}

Let $(1|1302)\in \mathbb{Z}_{2}^{1}\times \mathbb{Z}_{4}^{4}.$ Then this generates the code
 $$\mathcal{C}=\{(0|0000),(1|1302),(0|2200), (1|3102)\}.$$
 This code is a relative two-weight additive code $\mathcal{C}(4,5)$ and the quaternary part of the vector has $2$ units.
 \begin{theorem}
  Let $\mathcal{C}(m_{1},m)=\langle(u|v)\rangle$ be a relative two-weight additive code in $\mathbb{Z}_{2}^{\alpha}\times\mathbb{Z}_{4}^{\beta}$ to the subcode 
  $\mathcal{C}_{1}$ and $u$ has $l$ unit positions
  $v$ has $k$ unit positions and $s$ zero divisor positions, then $m_{1}=2k$ and $m=l+2s+k.$
 \end{theorem}
\begin{proof}
 Let $ x \neq 0$ in $\mathcal{C}_{1}.$ Given that $\mathcal{C}=\langle (u|v)\rangle$ and $\mathcal{C}_{1}\subset \mathcal{C},$
  implies that $x=2(u|v).$ Otherwise, $\mathcal{C}=\mathcal{C}_{1}.$ 
  Therefore, $wt(x)=wt(2(u|v))=wt((0|2v))=m_{1}.$
  Since $v$ contains $k$ unit positions, $2v$ contains $k$ zero divisor positions. 
  Therefore, $wt((0|2v))=2k.$ and hence $m_{1}=2k.$
  
  Let $(u|v)\in \mathcal{C}\setminus \mathcal{C}_{1},$ then $wt((u|v))=wt_{H}(u)+wt_{H}(\phi(v))=l+wt(\phi(v)).$  
  Since $v$ contains $k$ unit positions and $s$ zero divisor positions, 
 $\phi(v)$ contains $k +2s$ unit positions. 
 Therefore, $wt(\phi(v))=k+2s$
 and $wt((u|v))=l+k+2s.$ Hence $m=l+k+2s.$
\end{proof}

\section{Equivalence and Automorphism Groups of an Additive Codes}

 \begin{definition}
 Let $\mathcal{C}_{1}$ and $\mathcal{C}_{2}$ be two additive codes in $\mathbb{Z}_{2}^{\alpha}\times \mathbb{Z}_{4}^{\beta}.$
  We say that $\mathcal{C}_{1}$ and $\mathcal{C}_{2}$ are permutation equivalent if there exists $(\sigma, \tau)\in S_{\alpha}\times S_{\beta}$
  such that
  \begin{equation*}
  \mathcal{C}_{2}=\left\{(x_{\sigma(1)},x_{\sigma(2)},\cdots,x_{\sigma(\alpha)},y_{\tau(1)},y_{\tau(2)},\cdots,y_{\tau(\beta)}) \mid 
  (x_{1},x_{2},\cdots,x_{\alpha},y_{1},y_{2},\cdots,y_{\beta})\in \mathcal{C}_{1}\right\}
  \end{equation*}
  where $S_{\alpha}$, $S_{\beta}$ are the Symmetric groups on $\alpha$ and $\beta$ symbols, respectively.
  \end{definition}
  In other words, for every $\sigma\in S_{\alpha} \text{ and } \tau\in S_{\beta},$ the map
 $f_{\sigma,\tau}:\mathbb{Z}_{2}^{\alpha}\times \mathbb{Z}_{4}^{\beta}\rightarrow \mathbb{Z}_{2}^{\alpha}\times \mathbb{Z}_{4}^{\beta}$ 
 defined by
 $f_{\sigma,\tau}(x_{1},x_{2},\cdots,x_{\alpha},y_{1},y_{2},\cdots,y_{\beta}) =  
 (x_{\sigma(1)},x_{\sigma(2)},\cdots,x_{\sigma(\alpha)},y_{\tau(1)},y_{\tau(2)},\cdots,y_{\tau(\beta)})$
 for all $(x|y) \in \mathbb{Z}_{2}^{\alpha}\times \mathbb{Z}_{4}^{\beta}$  induces an isomorphism 
 from $\mathcal{C}_{1}$ onto $\mathcal{C}_{2}.$
 
 Equivalently, two additive codes $\mathcal{C}_{1}$ and $\mathcal{C}_{2}$ in  $\mathbb{Z}_{2}^{\alpha}\times \mathbb{Z}_{4}^{\beta}$ are permutation
 equivalent if there are permutation matrices $P_{\alpha}$ and $P_{\beta}$ such that $f_{\sigma,\tau}(x,y)=(xP_{\alpha},yP_{\beta})$ gives a bijection 
 of $\mathcal{C}_{1}$ onto $\mathcal{C}_{2}$ where $P_{\alpha}$ is a permutation matrix of order $\alpha$ and $P_{\beta}$ is a permutation matrix
  of order $\beta.$ It is denoted as $\mathcal{C}_{1}\sim \mathcal{C}_{2}.$ In fact, this relation is an equivalence relation.
  
 \begin{example}
  Let $\mathcal{C}_{1}=\langle (10|31)\rangle=\{(00|00),(10|13),(00|22),(10|31)\}$ and 
  
  $\mathcal{C}_{2}=\langle (01|31)\rangle=\{(00|00),(01|31),(00|22),(01|13)\}.$ 
  Then $\mathcal{C}_{1}\sim \mathcal{C}_{2}.$
 \end{example}
\begin{example}
 Let $\mathcal{C}_{1}=\langle (101|121)\rangle=\{(000|000),(101|121),(000|202),(101|323)\}$ and 
 
 $\mathcal{C}_{2}=\langle (110|112)\rangle
 =\{(000|000),(110|112),(000|220),(110|332)\}.$ Then $\mathcal{C}_{1}\sim \mathcal{C}_{2}.$
\end{example}

%
%
  Let $\mathcal{C}$ be an additive code in $\mathbb{Z}_{2}^{\alpha}\times \mathbb{Z}_{4}^{\beta}.$ 
  Then the permutation automorphism group of $\mathcal{C}$ is defined to be the set of 
  all $(\sigma,\tau) \in S_{\alpha}\times S_{\beta}$
  such that $(x_{\sigma(1)},x_{\sigma(2)},\cdots,x_{\sigma(\alpha)},y_{\tau(1)},y_{\tau(2)},\cdots,y_{\tau(\beta)})\in \mathcal{C}
 \text{ for all } (x_{1},\cdots,x_{\alpha},y_{1},\cdots,y_{\beta})\in \mathcal{C}$ and is denoted by  $PAut(\mathcal{C}).$
 
 Clearly, $PAut(\mathcal{C})$ is a subgroup of $S_{\alpha}\times S_{\beta}.$
 
 \begin{example}
   Let 
  $\mathcal{C}$ be the code generated by $(1010|1213),$ then the code \\$\mathcal{C}
  =\{(0000|0000),(1010|1213),(0000|2022),(1010|3231)\}$
  is an additive code in $\mathbb{Z}_{2}^{4}\times \mathbb{Z}_{4}^{4}$ and  
  $PAut(\mathcal{C})=\{(e_{4},e_{4}),((13),e_{4}),((24),e_{4}), (e_{4},(13))\}\subseteq 
  S_{4}\times S_{4}$ where $e_4$ is the identity element of $S_4.$
 \end{example}

 \begin{example}
  Let   $\mathcal{C}=\langle(10|11),(11|31)\rangle,$ then the code is $\{(00|00),(10|11),(10|33),\\(01|02),(01|20),(00|22),(11|13),(11|31)\}$ 
  and $PAut(\mathcal{C})=\{(e_{2},e_{2}),(e_{2},(12))\}\subseteq S_{2}\times S_{2}$ where $e_2$ is the identity element of $S_2.$
 \end{example}
\begin{proposition}
 Let $\mathcal{C}=\langle(u|v)\rangle $ be an additive code in $\mathbb{Z}_{2}^{\alpha}\times \mathbb{Z}_{4}^{\beta}.$ If $u$ has $l$ unit positions
 and $v$ has $k$ zero divisior positions, $s$ zero positions, $q$ $1's$ positions and $r$ $3's$ positions, 
 then $|PAut(\mathcal{C})|=(l!+(\alpha-l)!-1) (k!+s!+q!+r!-3).$
\end{proposition}
\begin{proof}
 Let  $u=(u_1,u_2,\cdots,u_{\alpha})\in \mathbb{Z}_{2}^{\alpha}$ and let 
 $v=(v_1,v_2,\cdots,v_{\beta})\in\mathbb{Z}_{4}^{\beta}.$
 
 For fixing $v,$ if we permute $l$ unit positions of $u,$ then we get $l!$ permutations $\sigma \in S_{\alpha}$ such that 
 $f_{\sigma,e_{\beta}}(u,v)\in \mathcal{C}$ and if we permute $\alpha-l$ zero positions of $u,$ we get $(\alpha-l)!$ permutations 
 $\sigma\in S_{\alpha}$ such that $f_{\sigma,e_{\beta}}(u,v)\in \mathcal{C}.$ Since $(e_{\alpha},e_{\beta})$ is in both 
 collection, for fixing $v,$ there are $l!+(\alpha-l)!-1$ permutations $\sigma \in S_{\alpha}$ such that 
 $f_{\sigma,e_{\beta}}(u,v)\in \mathcal{C}.$
 
 Similarly, for fixing $u,$ if we permute $k$ zero divisor positions, $s$ zero positions, $q$ $1$ positions and $r$ $3$ positions, then we get
 $k!,s!,q!$ and $r!$ permutations $\tau \in S_{\beta},$ respectively such that 
 $f_{e_{\alpha},\tau}(u,v)\in \mathcal{C}.$
 Since each collection has $(e_{\alpha},e_{\beta}),$  there are $k!+s!+q!+r!-3$ permutations 
 $\tau \in S_{\beta}$   such that $f_{e_{\alpha},\tau}(u,v)\in \mathcal{C}.$ 
 
 Therefore, for each $\sigma \in S_\alpha$ while fixing v, there are $k!+s!+q!+r!-3$ permutations $\tau \in S_{\beta}$  such that 
  $f_{\sigma,\tau}(u,v)\in \mathcal{C}$ and hence there are $(l!+(\alpha-l)!-1)(k!+s!+q!+r!-3)$ pair of permutations 
 $(\sigma,\tau)\in S_{\alpha}\times S_{\beta}$ such that  $f_{\sigma,\tau}(u,v)\in \mathcal{C}.$ 
 
 That is, $PAut(\mathcal{C})$ has $(l!+(\alpha-l)!-1)(k!+s!+q!+r!-3)$ elements.

\end{proof}

\begin{example}
 Let $\mathcal{C}=\langle (1101|1231)\rangle=\{(0000|0000),(1101|1231),(0000|2022),\\(1101,3213)\}.$
 Here $l=3,k=1$ and $s=0.$ By the above theorem, $|PAut(\mathcal{C})|=(3!+(4-3)!-1)(1!+0!+2!+1!-3)=(6)(2)=12.$
 
The group $PAut(\mathcal{C})=\{(e_{4},e_{4}),((12),e_{4}),((24),e_{4}),((14),e_{4}),((124),e_{4}),((142),e_{4}),\\
 (e_{4},(14)),((12),(14)),((24),(14)),((14),(14)),((124),(14)),((142),(14))\}\subseteq S_{4}\times S_{4}.$
\end{example}

An additive code $\mathcal{C}\subseteq \mathbb{Z}_{2}^{\alpha}\times\mathbb{Z}_{4}^{\beta}$ is called
a $\it{\mathbb{Z}_{2}\mathbb{Z}_{4}\text{-}additive\, cyclic\, code}$ if for every

 $(a_{0},a_{1},\cdots,a_{\alpha-1},b_{0},b_{1},\cdots,b_{\beta-1})\in \mathcal{C}$ $\Rightarrow \,(a_{\alpha-1},a_{0},
 \cdots,a_{\alpha-2},b_{\beta-1},b_{0},\cdots,b_{\beta-2})\in \mathcal{C}.$

\begin{theorem}
 Let $\mathcal{C}$ be an additive code in $\mathbb{Z}_{2}^{\alpha}\times \mathbb{Z}_{4}^{\beta}.$ Then $\mathcal{C}$ is cyclic iff 
 $PAut(\mathcal{C})= S_{\alpha}\times S_{\beta}.$
\end{theorem}

\begin{example}
 Let $\mathcal{C}=\langle (1111|3333)\rangle \subseteq \mathbb{Z}_{2}^{4}\times \mathbb{Z}_{4}^{4}$ be an additive code.
 Then $PAut(\mathcal{C})=S_{4}\times S_{4}.$
\end{example}
\section{Conclusion}
In this paper, we study a relative two-weight $\mathbb{Z}_{2}\mathbb{Z}_{4}$-additive codes. It is shown that the Gray image of a two-distance
$\mathbb{Z}_{2}\mathbb{Z}_{4}$-additive code is a binary two-distance code and that the Gray image of a relative two-weight
$\mathbb{Z}_{2}\mathbb{Z}_{4}$-additive code, with nontrivial binary part, is a linear binary relative two-weight code. 
The structure of relative two-weight $\mathbb{Z}_{2}\mathbb{Z}_{4}$-additive codes are described. 
Finally, we discussed permutation automorphism group of  $\mathbb{Z}_{2}\mathbb{Z}_{4}$-additive codes.




\end{document}